\documentclass[11pt]{article}

\usepackage[inline]{enumitem}
\usepackage{amsmath}
\usepackage{amsthm}
\usepackage{caption}
\usepackage{xspace}

\usepackage{tikz}
\usetikzlibrary{shapes,snakes}
\usetikzlibrary{positioning}
\usetikzlibrary{arrows.meta}

\usepackage[breaklinks=true,
            colorlinks = true,
            linkcolor = blue,
            urlcolor  = blue,
            citecolor = blue,
            anchorcolor = blue]{hyperref}
\usepackage{breakurl}           

\newcommand{\mailto}[1]{\href{mailto:#1}{\nolinkurl{#1}}}

\title{Bootstrapping a stable computation token}
\author{
Jason Teutsch\\
\emph{Truebit}\\
\mailto{jt@truebit.io}
\and
Sami M\"{a}kel\"{a}\\
\emph{Truebit}\\
\mailto{mrsmkl@gmail.com}
\and
Surya Bakshi\\
\emph{UIUC/Truebit}\\
\mailto{sbakshi3@illinois.edu}
}

\theoremstyle{remark}
\newtheorem*{claim}{Claim}
\theoremstyle{theorem}
\newtheorem*{prop}{Proposition}

\newcommand{\TRU}{\textsf{TRU}\xspace}
\newcommand{\CPU}{\textsf{CPU}\xspace}
\newcommand{\GPU}{\textsf{GPU}\xspace}
\newcommand{\DAO}{\textsf{DAO}\xspace}
\newcommand{\USD}{\textsf{USD}\xspace}
\newcommand{\XYZ}{\textsf{XYZ}\xspace}
\newcommand{\X}{\textsf{X}\xspace}
\newcommand{\Y}{\textsf{Y}\xspace}

\begin{document}

\maketitle

\begin{abstract}
We outline a token model for Truebit, a retrofitting, blockchain enhancement which enables secure, community-based computation.  The model addresses the challenge of stable task pricing, as raised in the Truebit whitepaper, without appealing to external oracles, exchanges, or hierarchical nodes.  The system's sustainable economics and fair market pricing derive from a mintable token format which leverages existing tokens for liquidity.  Finally, we introduce a governance layer whose lifecycles culminates with permanent dissolution into utility tokens, thereby tending the network towards autonomous decentralization.
\end{abstract}

\section{Initializing Truebit}

Bitcoin, from its outset, leveraged an egalitarian distribution method enabling anyone 
with an Internet-connected computer to obtain its native tokens~\cite{Nak08}.  The system's automated mining process indiscriminately, albeit probabilistically, distributes digital wealth in the form of bitcoins to individuals who run a Bitcoin client on their local machines.  While the practicalities of scaling, demand, and human nature introduced some complexities beyond those conceived in Bitcoin's original specifications \cite{bitcoin}, the system's remarkable minimization of \emph{politics}, or complex formal dependencies on the actions or judgements of privileged nodes, permit the underlying protocol to function under simple, mathematical assumptions.

Tokens born from and regulated by smart contracts, in contrast to Bitcoin's Nakamoto consensus,  present unique bootstrapping challenges with respect to both adoption and distribution.  Indeed, Bitcoin's value proposition of ``generate your own cash and then spend it'' doesn't translate easily into systems where intended \emph{consumers} must supply such cash.  We witness an abundance of miners and stakers offering computation power and capital in the blockchain space, however the corresponding consumers for these services, when distinct from the service providers, remain far less ubiquitous.  Some protocols, like Livepeer's MerkleMine~\cite{merklemine}, have dispensed tokens through computational work at the smart contract layer, yet sustainable distribution through apolitical function lingers elusively on blockchain's horizon.  In light of the relatively low demand for decentralized services, we concentrate on an economic design which minimizes friction and politics for consumers without sacrificing security.

Consumers generally find convenience in predictable pricing as acquisition of an asset often does not coincide with its consumption.  Consider a pilot who purchases sufficient aircraft fuel for a trip from Los Angeles to Tokyo and takes off.  Halfway through his journey, the price of fuel increases by 20\%. Consequently, a corresponding $1 - 1/(120\%)$ fraction of the pilot's remaining store vanishes from the gas tank, inconveniently diverting his course to Hawaii.  While the physical world may prevent this particular scenario from occurring, the volatile world of cryptocurrency consumption provides no such guarantees.  The pilot in this example requires a fixed amount of fuel, not a stable amount of fuel relative to the US dollar (\USD) (compare with \cite{maker} and \cite{tether}).  This example suggests the need for a kind of affordable, stable token independent of \USD (Section~\ref{sec:basic}). Both Truebit's stable token and fiat currency may correlate with the price of electricity (Section~\ref{sec:pricingtasks}). We highlight that the Truebit protocol model assumes no distinguished authority nodes and, as such, achieves not only a trustless computation system but a decentralized one based on simple security assumptions and hierarchy-free pricing.

Any new network which requires consumers to pay for services with a token for which they \emph{a priori} lack access faces a distribution problem.  This fundamental initialization challenge exists even for \emph{mineable} smart contract-based tokens, such as those used in Truebit~\cite{TR17}.  Some blockchain projects politically circumvent this utility dilemma through \emph{premining}, or initial distribution to a select group of individuals or institutions, however a private premine alone does not transform the system into a public good.  Sections~\ref{sec:pricingtasks}, \ref{sec:whitelistinggame}, and \ref{sec:governancegame} describe premining alternatives which leverage existing liquid tokens for distribution.  This technique reduces friction for consumers, who use assets readily available to them, while offering a potential source of revenue for project management and enhanced collaboration.

The governance game, as described in Section~\ref{sec:governancegame}, determines in the short run tokens for use in bootstrapping and in the long run incentives for those holding governance tokens to convert them into utility tokens.  Upon conversion of all governance tokens, a fully decentralized, yet upgradable system persists (Section~\ref{sec:upgrade}).  Teutsch and Reitwie{\ss}ner predicted some use cases for Truebit in early 2017~\cite[Section~7]{TR17}.  Armed with a modern imagination, we now bring these ideas to a practical test.

\section{Protocol review} \label{sec:review}
This exposition describes a deployment method for the Truebit protocol~\cite{TR17}, a blockchain enhancement which enables smart contracts to securely execute larger computations than the minimal gas limit permits~\cite{LTKS15}.  \emph{Task Givers} submit computational tasks, while \emph{Solvers} and \emph{Verifiers} ensure correct results in exchange for token rewards.  Each task may have several Verifiers but only one Solver.  The protocol runs on a unanimous consensus protocol among all Verifiers; there are no privileged or distinguished nodes. We shall largely treat the Truebit protocol as a black box, however the interested reader may refer to the project whitepaper \cite{TR17} and code \cite{truebitos} for full details.   While acquaintance with these specifications may prove useful, we aim to keep the present discussion streamlined and self-contained.

\begin{figure}[t]
\begin{center}
\begin{tikzpicture}
  [shorten <=0pt,shorten >=0pt,>={Stealth[length=7pt]},thick=3pt]
   \tikzstyle{contract} = [rectangle, minimum size=1cm, inner sep=10pt, outer sep=1ex, font={\normalsize \tt}, fill=none, thick, draw]
   \tikzstyle{subcontract} = [rectangle, inner sep=2pt, outer sep=3pt, font={\normalsize \tt}, fill=none, thick, dashed, draw]
   \tikzstyle{entity} = [dashed, star, star points=2, star point ratio=0.8, minimum size=1.7cm, inner sep=4pt, outer sep=1ex, font={\normalsize \sf}, fill=none, thick, draw]

   \node[contract](protocol) at (0,0) {\Large Truebit protocol};
   \node[entity] (tasking) at (-3,3) {\large Tasking};
   \node[entity] (staking) at (3,3) {\large Staking};
   \node[entity] (reward) at (0,-3) {\large Reward};
   
   \draw[->] (tasking) to node[yshift=-1ex,left] {\begin{tabular}{c}\CPU \\ \X \end{tabular}} (protocol);
   \draw[->] (protocol) to node[right] {\TRU} (reward);
   \draw[->] (staking) to node[yshift=-1ex, right] {\begin{tabular}{c}\TRU \\ \Y \end{tabular}} (protocol);
\end{tikzpicture}
\caption{Simplified token functions in Truebit.  \CPU  and \TRU are primary tokens (Sections~\ref{sec:basic}), whereas ``\X'' and ``\Y'' denote indeterminate, or ``variable,'' tokens.} \label{fig:functions}
\end{center}
\end{figure}
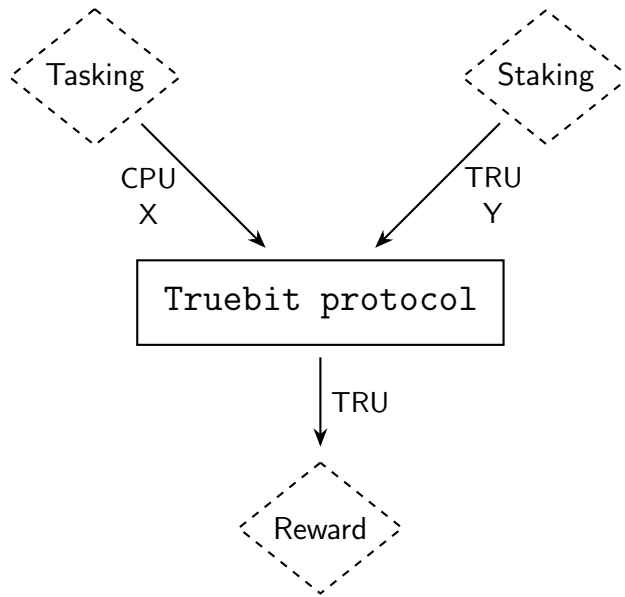

Tokens have three functions in Truebit: paying for tasks, staking to participate, and rewarding Solvers and Verifiers (Figure~\ref{fig:functions}).  Just as the identities of Task Givers, Solvers, and Verifiers may either be disjoint or overlap, so too may the actual \emph{tasking}, \emph{staking}, and \emph{reward} tokens.  The operating model, depicted in Figure~\ref{fig:model}, enables the possibility of disjoint tokens by minting rewards.

Solvers and Verifiers in Truebit must stake token deposits in order to participate in the network.  This reserves financing for dispute resolution, in case some Verifier disagrees with a Solver's solution.  Staking penalties, enforced in the framework of a \emph{verification game}, disincentivize intentionally false or spammy solutions from Solvers as well as false alarms from Verifiers.  When~$n$ Verifiers verify a task, then total token reward decreases by a multiplicative factor of $2^{n-1}$, an \emph{exponential dropoff} which dissuades Sybil attacks, and the protocol distributes the reward evenly among all participating Verifiers.  In Section~\ref{sec:random}, we discuss an alternative mechanism with linear payouts rather than exponential dropoff.  The Truebit protocol as described in the whitepaper uses a probabilistic reward scheme which pays Verifiers for discovering bugs.  This \emph{forced error} mechanism occasionally rewards Solvers for providing wrong answers and ensures that Verifiers get some rewards.  The whitepaper describes a \emph{jackpot repository} which cumulatively stores a portion of fees paid from each task in order to fund reward payouts in the case of forced errors.

The current Truebit implementation~\cite{truebitos} can process task data from various sources, including IPFS \cite{ipfs}.  If the raw data corresponding to the input IPFS hash is not made publicly available~\cite{Teu17}, then Truebit effectively processes private data via some subset of Solvers and Verifiers.  A task which makes use of such a private input is called a \emph{private task}.  Goddard \cite{god16} first pointed out this innate privacy feature of interactive verification, and the Arbitrum project later popularized it~\cite{KGCWF18}.  The current Truebit implementation~\cite{truebitos} processes both public and private inputs uniformly, cost-effectively, and faster than the authors of the original whitepaper could have imagined \cite{truebitjit}.

\section{Basic token operations} \label{sec:basic}

We express the components of the token model through a series of self-contained but cumulative updates to the original Truebit whitepaper~\cite{TR17}.  As indicated in Figure~\ref{fig:functions}, the system primarily relies on \TRU tokens for staking and uses \CPU for task payments.  A hard-wired preference for \TRU rewards aids predictable task pricing, and the multi-token model achieves predictable economic effects by isolating the properties and functions of each system component.

The Truebit whitepaper describes an \emph{underpricing attack} in which a Task Giver issues a task at a negligible price so that \emph{rational miners}, or Solvers and Verifiers motivated primarily by token rewards, lack incentive to solve it  \cite[Section~5.5]{TR17}.  The Task Giver then solves the task himself, gets a bogus answer, and breaks the system's security.  More simply, a Task Giver might, if left to his own devices, underprice a task simply to cut costs.  Hence the protocol cannot allow Task Givers to price tasks under the assumption of rational miners.  For similar reasons, due to potential Sybil identities, individual Solvers and Verifiers cannot securely set task costs for other participants.

In light of these requirements, Truebit relieves Task Givers from pricing responsibilities by fixing the cost for one computation cycle at one \CPU.  Thus Task Givers holding \CPU may consume them at any time without exposure to price variation, while Solvers and Verifiers absorb potential market fluctuations through the \TRU token (Section~\ref{sec:pricingtasks}).  This two token model decouples stability and market functions in a way reminiscent of GasToken~\cite{gastoken} and without the protocol or governance system determining inflationary or deflationary rates; we assume that Task Givers holding \CPU condone indeterminate token supply so long as it does not inhibit their ability to issue tasks and obtain correct computational results.  

As a guiding design principle, the value of tokens paid into the Truebit protocol must not exceed the value of tokens paid out as rewards, lest Task Givers accumulate \TRU rewards for spamming the network by solving their own tasks.  Indeed, one can efficiently pass entire task rewards to oneself via private tasks as detailed in Section~\ref{sec:pricingrewards}.  Regardless of the quantity of \TRU received as a reward, the reward received per computation step is always equivalent to one~\CPU according to the pricing provided by the reward recipient.

We expect that many Verifiers participate in the network in order to guarantee correctness of their own tasks, however these same Solvers and Verifiers may add security to other Task Givers by leaving their Verifier or Solver client running between solving their own tasks.  Thus participants who help the Truebit community may effectively use the system at negligible cost (ignoring network gas and real electricity expenses).

\subsection{Pricing rewards} \label{sec:pricingrewards}

Rational miners wish to receive no less than the market rate for their services.  We elaborate on Truebit's staking method in order to meet this requirement.  Each participating Solver/Verifier, $V$, does the following prior to solving tasks.
\begin{enumerate}[label=\arabic*.]
\item $V$ specifies a \emph{local price} per computation step in \TRU, $p$, to be paid to Solvers and to Verifiers, and
\item stakes \emph{both} \CPU and \TRU tokens in a ratio of 1 to $p$.
\end{enumerate}
For conceptual clarity, we have denominated reward payouts with respect to number of computation steps, however in practice the local price might represent \emph{expected} price per computation step when taking into account the sporadic, cumulative nature of jackpot payouts.

Whenever a Solver/Verifier $V$ earns a reward, the protocol \emph{mints} to $V$ the \TRU reward allocation at the median price among all local prices, while the protocol \emph{burns}, or permanently destroys, the Task Giver's corresponding payment.  Median pricing, as opposed to localized pricing, avoids a potential tragedy of the commons wherein parties maximize their own rewards at the detriment of the tasking conversion price discussed in the next section.  The system expects the local price to equal the ``true'' cost of a computation step in \TRU, but what's to stop $V$ from collecting unbounded rewards by choosing $p$ as large as possible?  The protocol incentivizes each Solver/Verifier $V$ to provide correct market pricing through \emph{Staking Monitors} who can \emph{swap} tokens against $V$'s stake at $V$'s local price $p$.  In other words, anyone holding a \CPU token can exchange it for $p$~\TRU tokens against $V$'s deposit, less a fixed, \emph{universal staking conversion fee} of, say, $20\%$ to dissuade swapping against reasonable pricing.  In more detail, a \emph{staking swap} proceeds as follows.
\begin{enumerate}[label=(\alph*)]
\item A Staking Monitor~$M$ targets the deposit of some Solver/Verifier~$V$ with local price~$p$.

\item $M$ submits $n$ \CPU (resp.\ \TRU) to the Staking contract.

\item Assuming $V$ had staked at least $np$ \TRU (resp.\ $n/p$ \CPU):
\begin{enumerate}[label=\roman*.]
\item $M$ obtains $(1-c)np$ \TRU (resp.\ $[1-c]n/p$ \CPU) from $V$'s stake, and 
\item $V$ retains the $n$ \CPU (resp.\ \TRU) submitted by $M$.
\end{enumerate}
\end{enumerate}
The protocol requires a \emph{price bonding delay} of, say, 24 hours between the time that a Solver/Verifier initially commits or updates a price per computation step and the time that the Solver/Verifier can participate and receive rewards.  This delay enables Staking Monitors to confirm the Solver/Verifier's local price before it applies in the system.  At all times, the Solver/Verifier must have stake in the system in order to participate, however stake may be withdrawn at any time, including during the price bonding delay period.  This allows Solvers and Verifiers to protect their deposits against potential price volatility.

In Section~\ref{sec:pricingtasks}, we shall construct a more cost efficient means of converting from \TRU to \CPU tokens; staking swaps are only preferable in cases of outlier local pricing.  Thanks to \TRU's designation as reward token, one can achieve the reverse exchange direction, from \CPU to \TRU, via the method of private tasks.

\paragraph{Private tasks.}
In cases where only a hash of the input data appears on-chain, some nodes may never have access to a task's raw input.  Private tasks offer rudimentary but effective security by obscurity and arise whenever the input to a Truebit task points to data outside the blockchain or permanent contract storage, e.g.\ through reads and writes to IPFS.  For example, a Task Giver can supply an IPFS hash as input to Truebit's on-chain filesystem without actually uploading any public data to IPFS.

A Task Giver who submits a private task can effectively guarantee that he will be the only Solver and Verifier by ensuring that no one else has access to the raw input data.  As no other party can correctly perform the requested computation, the Task Giver could win any verification game~\cite{TR17} that arises and therefore destroy the deposit of any opposing party.  Thus, by issuing a private task and also solving it, the Task Giver can efficiently burn \CPU and in its place receive \TRU as reward.

\paragraph{Jackpot spiral.}
The astute reader may notice that minting \TRU rewards ``on-the-fly,'' permits periodic payment to Verifiers without requiring the existence of a jackpot repository to accumulate the necessary funding.  In the Truebit whitepaper~\cite{TR17}, the finite quantity of tokens stored in a jackpot repository bounds the maximum incentive for Verifiers and hence the size of the largest, processable, secure computation.  With \TRU minting, in contrast, the Truebit protocol can, in theory, reward Solvers and Verifiers for performing tasks of any size.  Since the construction above vanquishes the incentive bounding parameter, the \TRU token itself now becomes a scalability solution.

Viewed another way, the design mitigates risks by eliminating politics of replenishing and growing the jackpot repository.  Karen Teutsch pointed out that a finite jackpot repository without active, altruistic management is susceptible to \emph{jackpot spiral}.  Suppose that the Truebit protocol were to use a finite jackpot repository instead of a mintable \TRU.  If the value of \TRU according to the pricing contract were to decrease for any reason, or if the jackpot repository's drops due to normal variance, then a previously executable task could become out of reach for the Truebit protocol.  In all likelihood, this means that some DApp which relies on Truebit would stop working, which causes the price of \TRU to drop further, which causes additional DApps to fail, which causes the price of \TRU to drop further, resulting in an undesirable, hyperinflationary feedback loop.

\subsection{Tasking token economics} \label{sec:pricingtasks}

Solver/Verifiers may wish to redeem rewards by issuing tasks.    While Solvers and Verifiers receive rewards in \TRU, Task Givers pay for computations with \CPU.  The system therefore requires an efficient means of converting \TRU into \CPU.  Let us inspect the economics of this process.

The system prices \emph{tasking conversions} from \TRU in \CPU using the median over all bonded Solver/Verifiers' local prices (Section~\ref{sec:pricingrewards}).  In short, Solvers, Verifiers, and Task Givers can burn \TRU into a \emph{tasking conversion contract} which instantly mints back an equivalent quantity of \CPU tokens in the sense that one \CPU token always pays for one computation step.  More precisely, the protocol mints a quantity of \CPU equal to the number of \TRU tokens burned divided by the median local price.  Note that a smart contract can efficiently maintain this median price in contract storage.

We now turn our attention to the system's relationship with external tokens and fiat currency.  The protocol supports tasking conversions from tokens other than \TRU as illustrated in Figure~\ref{fig:model}.  A \emph{Holder} in possession of $\XYZ$ tokens, who wishes to obtain \CPU
\begin{enumerate}
\item stakes both \XYZ and \TRU in the tasking conversion contract, and
\item provides an exchange rate between \XYZ and \TRU.
\end{enumerate}
Analogous to the procedure for Staking Monitors (Section~\ref{sec:pricingrewards}), any Tasking Monitor can swap the Holder's deposit from \XYZ to \TRU or \TRU to \XYZ before the end of pricing bonding period.  If no Tasking Monitor swaps the Holder's deposit in the allotted time, then the tasking contract proceeds to convert the \XYZ deposit into \CPU at the  rate in Item~2 composed with the median price described in the previous paragraph.

The above construction assumes that the \XYZ token has sufficient liquidity that a Tasking Monitor can not only estimate the correct price for \XYZ but also has access to \XYZ tokens with which to swap.  Secondly, we assume that the Holder has some access to a modest number of \TRU tokens, either from friends, through a Uniswap contract \cite{uniswap}, or via whitelisted activities (Section~\ref{sec:bootstrapping}).  Lastly, the amount of \TRU posted must be enough to make conversion worthwhile to the Tasking Monitor, e.g.\ equal to the Verifier \TRU deposit for some standard task.  The total value of the \TRU deposited, however, need not equal the value of the \XYZ deposit.  For example, suppose that an \XYZ Holder wants to mint 100 \CPU tokens and that the Holder's declared exchange rate is 2 \XYZ per 1 \TRU.  The Holder deposits 400 \XYZ tokens but only 10 \TRU tokens into the external conversion contract.  A monitoring agent may then exchange 20 \XYZ for the the 10 \TRU tokens, or vice-versa.

We remark that an \XYZ holder may burn \XYZ tokens repeatedly into the tasking conversion, using the same \TRU deposit, to obtain additional \CPU tokens.
Since this operation increases the net supply of tokens in the system, care must be taken to ensure that the supply grows in a controlled manner.  We shall explore this matter further in Section~\ref{sec:whitelistinggame}.

\paragraph{The value of \CPU in fiat.}
While one \CPU token always pays for one Truebit task, the exchange rate between \CPU and \USD may fluctuate over time.  On one hand, a \CPU price increase may incentivize Task Givers to adjust their \CPU purchases, task consumptions, and \TRU reward accumulations, but it also attracts rational miners and therefore maintains system integrity.  A \CPU price decrease, on the other hand, requires more detailed analysis to justify profitability for Solvers and Verifiers.

A \TRU Holder has options to hodl, sell, convert to \CPU, or, following the latter case, issue a task.  Consider the following pair of basic economic strategies for an active Solver/Verifier.
\begin{description}
\item{\emph{Strategy 1.}} Convert \CPU as soon as it's earned.  Use a price per computation step close the median \emph{plus} 20\% in order to maximize the amount of \CPU obtained during conversion without attracting the attention of a Staking Monitor.

\item{\emph{Strategy 2}.} Hodl earned \TRU.  Use a price per computation step close the median \emph{minus} 20\% in order to drive the median price of \TRU relative to \CPU as high as possible without attracting the attention of a Staking Monitor.
\end{description}
Whether or not \TRU inflation occurs relative to \CPU depends on whether the majority of Solver/Verifiers select Strategy~1 or Strategy~2.

Below we describe reasonable conditions under which Strategy~2 is the long-run, dominant strategy for rational miners.  In a nutshell, the set of hodlers following Strategy 2 eventually ends up with all the tokens while others are spending them on tasks.  This decrease in market \CPU supply ultimately drives up the price of \CPU relative to \USD.
\begin{prop}
Suppose that initially there exist $n$ \CPU tokens, including the equivalent value in \TRU tokens according to the median price, and
\begin{enumerate}[label=(\alph*)]
\item at least $p$ fraction of Solver/Verifiers follow Strategy~2,
\item on average, no more than (the equivalent of) $x < p$ new \CPU tokens are generated during each epoch of $n$ tasks (as a result of either \TRU deflation or external token conversions).
\end{enumerate}
Then the number of tasks until the set of Solver/Verifiers following Strategy~2 hold all of the tokens is
\[
\frac{n}{p-x}.
\]
\end{prop}
\begin{proof}
The expected time, as measured in tasks for the set of Solvers following Strategy~2 to acquire $n$ tokens is $n/p$.  During this epoch, $xn/p$ new \CPU tokens (including \TRU equivalent) are created, and it takes $xn/p^2$ time to acquire them.  By iterating this calculation, we see that the time required to acquire all tokens converges to the geometric series
\[
\frac{n}{p} + \frac{xn}{p^2} + \frac{x^2n}{p^3} + \dotsb
= \frac{n}{p} \cdot \prod_{k=0}^\infty \left(\frac{x}{p}\right)^k
= \frac{n/p}{1-x/p}
= \frac{n}{p-x}
\]
whenever $0 \leq x/p < 1$.  If $x < 0$, then the expected time to collect all the tokens is bounded above by $n/p$.
\end{proof}
Due to the median pricing scheme for tasking conversions, at least half of Solver/Verifer nodes must follow the deflationary scheme from Strategy~2 in order for \TRU value to increase relative to \CPU.  Solvers and Verifiers must each weigh pros and cons in deciding local prices.  A median increase in \TRU value relative to \CPU effectively increases the available \CPU supply.  While an increase in \TRU price grants \TRU Holders greater purchase power and attracts Task Givers via reduced \USD-equivalent prices, the change also reduces \USD-equivalent rewards for Solvers and Verifiers.  Thus local pricing brings into play both short-term economic and long-term network effects.  We emphasize that the protocol's market values for \TRU and \CPU tokens neither appeal to external price oracles nor exchanges, nor do they depend on specific details of the reward mechanism (e.g.\ number of Verifiers per task).

\paragraph{Effect of exchange markets.}
We have argued that the Truebit token system functions as intended without exposure to external markets, however in general the protocol cannot guarantee that such markets will not materialize.  We now consider potential effects of liquid exchanges and argue that they do not disrupt the construction.  There are four possibilities depending on which subsets of the two tokens, \TRU and \CPU, are \emph{tradeable} on public exchanges.  Recall that a \TRU holder may convert tokens to \CPU by using the protocol's built-in options contract as opposed to trading \TRU on an \emph{exchange}.

We have already considered the case where neither \TRU nor \CPU is tradable.  Now let us assume a liquid market for both \TRU and \CPU, and assume that an exchange value exists for each token.  Let $\USD(\TRU)$ [resp.\ $\USD(\CPU)$] denote the exchange market value for \TRU [resp.\ $\CPU$], and let the number $r$ be the \emph{tasking conversion rate} such that $r$ \TRU's convert to 1 $\CPU$.
\begin{claim}
Assume a fixed tasking conversion rate~$r$.  Then $\USD(\CPU)$ tends towards $r \cdot \USD(\TRU)$.
\end{claim}
\begin{proof}
If ever $\USD(\CPU) > r \cdot \USD(\TRU)$, then rational actors will convert rather than trade, hence the market does not support this pricing.    
On the other hand, if $\USD(\CPU) < r \cdot \USD(\TRU)$, then rational actors will use exchanges rather than converting.  Task Givers continue to burn \CPU while no new tokens are created, hence the \CPU supply decreases,  and therefore $\USD(\CPU)$ increases.  The claim follows.
\end{proof}

The remaining two cases, where exactly one of \TRU or \CPU is tradable are less interesting, however we include them for completeness.  In either case, exchanges add a sell option, but \TRU and \CPU prices remain incomparable.  In case only \TRU is tradable, then \CPU can only be obtained through the task conversion contract.  Then \CPU's prescribed functionality persists, and the \TRU price provides little distraction.  On the other hand, a tradable \CPU might make Solvers and Verifiers more inclined to convert from \TRU into \CPU, in which case they can simply modify their local prices according to the logic described above.

\section{Bootstrapping the network} \label{sec:bootstrapping}

The \CPU/\TRU lifecycle steps described in the previous subsection presumes the existence of tokens, yet none are present at network initialization.  It remains to derive the circumstances of genesis.  We describe a bootstrapping method which conveniently minimizes friction for participation and takes advantage of flexible token supply.  Sections~\ref{sec:whitelistinggame} and~\ref{sec:governancegame} provide mutually compatible options for initiating and maintaining a supply of \TRU and \CPU.  The methods described herein result in net inflation.

\subsection{Staking with external tokens} \label{sec:whitelistinggame}

For the purposes of initial distribution, in the short run, the protocol may permit Solvers and Verifiers to stake other tokens in place of \TRU, using the same method outlined in Section~\ref{sec:pricingrewards}, and Task Givers to burn tokens other than \CPU as payment for issuing tasks under ad-hoc pricing.  While the protocol permits withdrawal of stakes, Solvers and Verifiers would still retrieve deposits in the same currency in which they originally staked.  In case the external tokens used already have a broad, liquid distribution, Truebit immediately becomes an accessible, public system.  In order to facilitate Solver and Verifier participation and legitimate use, the system restricts private tasks to those issued with \CPU.  Once a sufficient initial distribution has been established so that Task Givers, Solvers, and Verifiers can access small quantities of \TRU, this initial process would end.

Offering staking services opens a potential source of revenue with which to finance the development team's operations.  The team can benefit from various business transaction structures, including fiat currency payments, value-in-kind services, and investment.  Through this short-term program, external tokens enjoy greater demand through increased utility, visibility, and obliteration of supply via tasking.

Since the only way to obtain \TRU tokens is through Solver and Verifier rewards, the post-whitelist staking mechanism above enables an apolitical distribution of tokens devoid of management decisions, centralized trust, or passive income, in contrast to other common techniques such as exchanges, premines, initial coin offerings, \cite{TBB17}, or airdrops~\cite{merklemine}.  Because ``an airdrop may constitute a sale or distribution of securities''~\cite{secframework}, the apparent lack of investment contract in this distribution method may offer regulatory benefits.   Even for cases where an external token use is restricted or ``locked'' for regulatory compliance reasons, whitelisting it in Truebit can potentially increase its utility, thereby making the external token less like a security.

Following the remark made in Section~\ref{sec:pricingtasks} regarding external tasking conversions, the development team chooses relevant tokens \XYZ on the same blockchain as Truebit and may bound tasking and/or staking functionality in at least one of the following parameters: the number of \XYZ tokens useable per unit time, total time allowed for whitelisting \XYZ tokens, and total quantity of \XYZ whitelisted.  Figure~\ref{fig:model} illustrates the context for these governance interactions which, as we shall discuss in Section~\ref{sec:governancegame}, need not be controlled by a centralized entity.

\subsection{The governance game} \label{sec:governancegame}

Let us now consider a tokenized version of the governance mechanism outlined in Section~\ref{sec:whitelistinggame}.  Assume some initial distribution of governance tokens, called \DAO, with democratic voting power (e.g. \cite{aragonvoting}, \cite{daostack})  restricted to the following set of items:
\begin{enumerate}[label=\arabic*.]
\item whitelisting of variable tokens for use in the tasking conversion contract, and

\item assigning a maximum useable allotment of each such variable token.
\end{enumerate}
As \DAO tokens convert over time, the protocol roughly tends towards decentralization as the Truebit protocol's political hierarchy, along with the net \TRU/\CPU token supply variance, fully dissolve upon the ultimate \DAO conversion. On the other hand, the last remaining \DAO tokens, which could belong to anyone who acquired even a small holding, increasingly resemble unique souvenirs and may be slow to disappear.

The governance token, or \DAO, expressly does not manage Truebit protocol upgrades (see Section~\ref{sec:upgrade}), and in particular cannot assign \TRU or \CPU minting rights to any new smart contract.  We enumerate the \DAO token's crypto-idiosyncratic features which resemble the independent, concurrent material in \cite{Zur19}.
\newpage
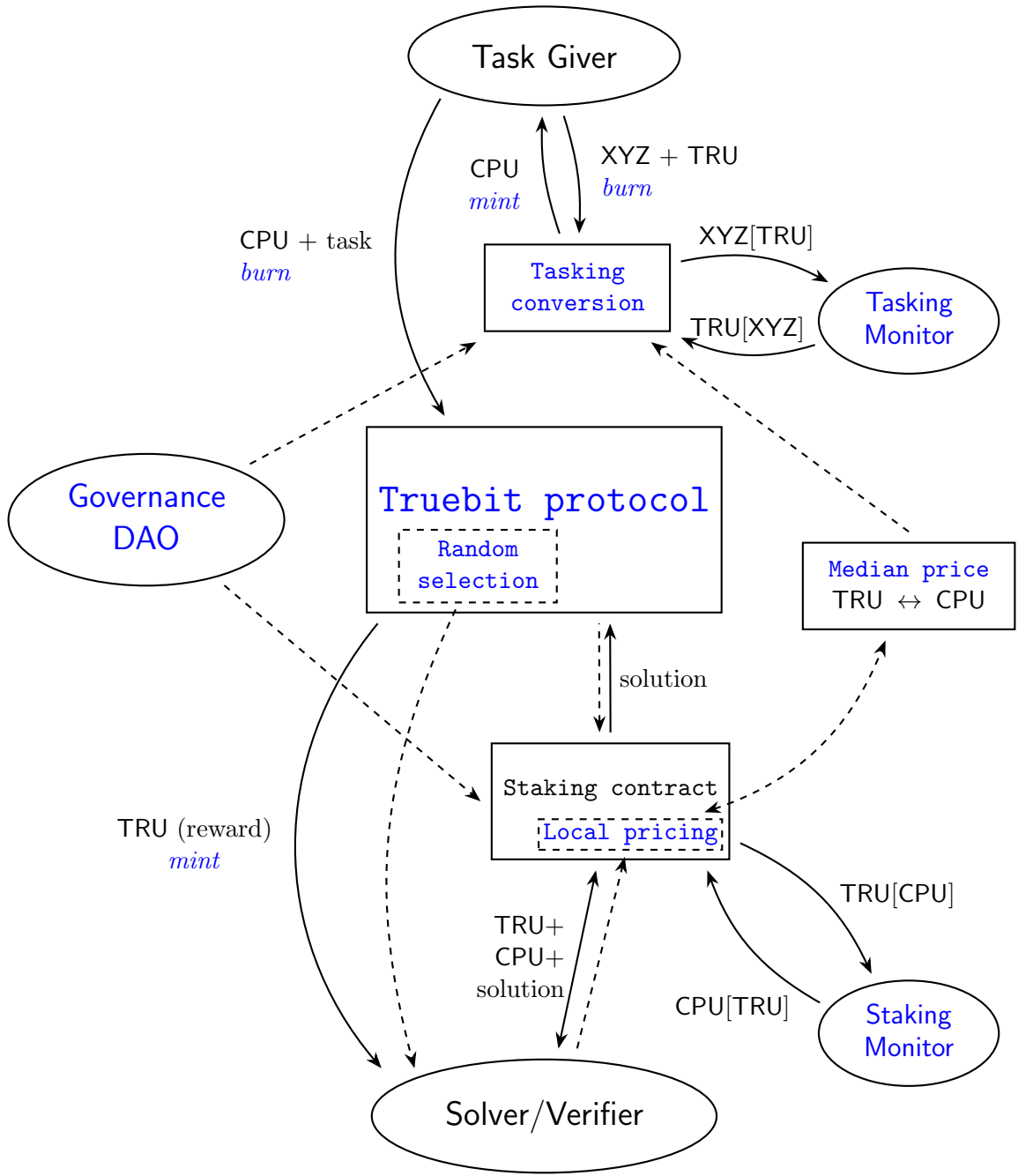
\begin{figure}[h!]
\vspace{-0.89in}
\hspace{-10ex}
\begin{tikzpicture}
  [shorten <=0pt,shorten >=0pt,>={Stealth[length=7pt]},thick=3pt]
   \tikzstyle{contract} = [rectangle, minimum size=1cm, inner sep=5pt, outer sep=1ex, font={\normalsize \tt}, fill=none, thick, draw]
   \tikzstyle{subcontract} = [rectangle, inner sep=2pt, outer sep=3pt, font={\normalsize \tt}, fill=none, thick, dashed, draw]
   \tikzstyle{entity} = [ellipse, minimum size=1cm, inner sep=10pt, outer sep=1ex, font={\normalsize \sf}, fill=none, thick, draw]

   \node[contract, minimum size = 2.8cm, label = {[yshift=1.5ex]center:\hyperref[sec:review]{\LARGE \tt Truebit protocol}}] (protocol) at (0,0) {\phantom{\hyperref[sec:review]{\LARGE Truebit protocol}}};
   \node[subcontract] (random) at (-1,-.7) {\hyperref[sec:random]{\begin{tabular}{c}{Random} \\selection\end{tabular}}};
   \node[entity] (taskgiver) at (0,7) {\Large Task Giver};
   \node[entity] (solver) at (0,-9) {\Large Solver/Verifier};
   \node[contract, minimum size= 1.75cm, label = {[yshift=1ex]center:\hyperref[sec:macroo]{\tt Staking contract}}] (staking) at (1,-4.25) {\phantom{\hyperref[sec:macro]{Staking contract}}};
   \node[contract] (median) at (5.5,-1) {\begin{tabular}{c}\hyperref[sec:pricingtasks]{Median price}\\ \TRU $\leftrightarrow$ \CPU\end{tabular}};
   \node[entity, inner sep=2pt] (governance) at (-6,0) {\Large \hyperref[sec:governancegame]{\begin{tabular}{c}Governance\\\DAO\end{tabular}}};
   \node[subcontract] (localpricing) at (1.3,-4.75) {\hyperref[sec:pricingrewards]{Local pricing}};
   \node[contract] (options) at (0.5,3.5) {\hyperref[sec:pricingtasks]{\begin{tabular}{c}Tasking\\ conversion\end{tabular}}};
   \node[entity, inner sep=2pt] (stakingobserver) at (5.5,-7.75) {\large \hyperref[sec:pricingrewards]{\begin{tabular}{c}Staking\\Monitor\end{tabular}}};
   \node[entity, inner sep=2pt] (taskingobserver) at (5.5,3) {\large \hyperref[sec:pricingtasks]{\begin{tabular}{c}Tasking\\Monitor\end{tabular}}};   
   
   \draw[->] ([xshift= 2ex]protocol.south west) to  [bend right = 40]  node[left] {\begin{tabular}{c}\TRU (reward) \\ \hyperref[sec:pricingrewards]{\textit{mint}} \end{tabular}}([xshift=-3ex] solver.north west);
   \draw[<->] (staking) to node[left,xshift=1ex] {\begin{tabular}{r}$\TRU+$\\$\CPU+$\\ solution\end{tabular}} (solver);
   \draw[->] (staking.north) to node[right] {solution} ([xshift = 6ex]protocol.south);
   \draw[<-,dashed] ([xshift = -1ex]staking.north) to ([xshift = 5ex]protocol.south);
   \draw[<->,dashed] (localpricing) to [bend right=25] (median);
   \draw[->,dashed] ([xshift=3ex]solver.north) to (localpricing);
   \draw[->,dashed] (governance) to (staking.west);
   \draw[->,dashed] (random) to [bend right = 20] (solver.north west);
   \draw[->,] (taskgiver.south west) to [bend right=30] node[left,shift={(0ex,0ex)}]{\begin{tabular}{l}$\CPU$ + task\\\hyperref[sec:pricingrewards]{\textit{burn}}\end{tabular}} (protocol);
   \draw[->,dashed] (governance) to (options);
   \draw[->] (taskgiver) to [bend left = 10] node[right]{\begin{tabular}{l}$\XYZ$ + $\TRU$\\\hyperref[sec:pricingtasks]{\textit{burn}}\end{tabular}} (options);
   \draw[->,dashed] (median.north) to [bend right = 0] (options);
   \draw[->] (options) to [bend left = 10] node[left, shift={(0ex,-1ex)}]{\begin{tabular}{c}\CPU\\ \hyperref[sec:pricingtasks]{\textit{mint}}\end{tabular}} (taskgiver);
   \draw[->] (stakingobserver) to [bend left = 20] node[below left, shift={(4.5ex,-3.5ex)}]{$\CPU[\TRU]$} ([xshift=-3ex]staking.south east);
   \draw[<-] (stakingobserver) to [bend right = 20] node[right, xshift=1ex]{$\TRU [\CPU]$} (staking);
    \draw[->] (taskingobserver) to [bend left = 20] node[above, yshift=.25ex]{$\TRU[\XYZ]$} (options);
   \draw[<-] (taskingobserver) to [bend right = 20] node[above, yshift=-.25ex]{$\XYZ[\TRU]$} (options);
   \end{tikzpicture}
\caption{Illustration of Truebit token model.  Solid lines indicate data and token transfers while dotted lines indicate control features.  ``\XYZ'' is placeholder symbol for external, whitelisted tokens, and square brackets adjacent to Monitors should be read as ``respectively.''  Governance actions for tasking and staking include whitelisting external tokens and fixing their maximum useable allotments.} \label{fig:model}
\end{figure}
\newpage

\noindent
\begin{enumerate}[label=(\alph*)]
\item Each \DAO token, at the token holder's discretion, is eligible for a one-time conversion into either \TRU or \CPU.
\item Each \DAO token increases in its governance power over time because conversion decreases the total supply of \DAO tokens.
\item Later conversions receive a higher percentage of \TRU or \CPU tokens.  More quantitatively, suppose that a DAO Holder has $p$ fraction of the original \DAO token supply and that $c$ fraction of the \DAO tokens have thus far been converted.  Without loss of generality, assume that the Holder converts to \TRU and that $N$ \TRU tokens currently remain.  Then the token holder receives
\begin{equation} \label{eqn:bonus}
f(p,c) = (p + 5c^2 p)\cdot N~\TRU,
\end{equation}
where ``$5c^2p$'' denotes the monotonic \emph{bonus} for holding \DAO long-term.
\end{enumerate}

Let consider what might happen if $p$ represented the fraction of currently remaining \DAO tokens rather than the fraction relative to the original supply of \DAO tokens.  As a simple illustration, assume $f(p) = pN$, with $N=100~\TRU$, and suppose that the entirety of \DAO holders decide to convert 50\% of their tokens to \TRU.  They get 50 \TRU, bringing the total \TRU supply to 150.  Now the \DAO holders again convert 50\%.  This time, they receive $50\% \cdot 150~\TRU = 75~\TRU$.  Hence, by converting 50\% and then 50\% again they receive a grand total of $$50~\TRU +75~\TRU= 125~\TRU,$$ whereas if they converted 100\% at first, they would only have received 100~\TRU.  Clearly obtaining an unbounded number of tokens from a finite amount of \DAO is undesirable.  We remark that the squared term ``$5c^2$'' in Equation~\ref{eqn:bonus} incentivizes long-term holding by back-loading the conversion bonus, while the number ``5,'' on the other hand, is a somewhat arbitrary constant.

Don Gossen acutely noted that one could effectively apply the governance conversion process above without a whitelisting feature.

\subsection{The upgrade game} \label{sec:upgrade}
Over time, network software may require amendments due to new fixes or features.  While governance (Section~\ref{sec:governancegame}) enables some foreseeable modifications, such as additions to the staking whitelist, others may require changes to the underlying, immutable smart contracts.  In the long run, the system requires a versatile update mechanism which preserves network effects and minting functionality without resorting to management decisions from a central authority.  To this end, we propose a simple framework which invites new fixes and features without breaking legacy contract code.
\begin{enumerate}
\item A new version of Truebit protocol smart contracts is deployed.

\item The new contracts accept both \CPU and \TRU tokens but mint new $\CPU'$ and $\TRU'$ tokens for conversions and rewards respectively.  The latter maintain tasking and staking functionality alongside legacy tokens.

\item \DAO tokens can convert to the new $\CPU'$ and $\TRU'$ using the new contracts as well as to legacy tokens using the old contracts (Section~\ref{sec:governancegame}).
\end{enumerate}
For purposes of greater governance diversity, a $\DAO'$ token might exist as well.  Upgrades can also introduce new governance rules, e.g.\ imposing a limit on the maximum number of simultaneous tokens available for the staking whitelist.  We remark that \TRU, \CPU, and \DAO holders all have incentive to follow this upgrade pattern which preserves the original features of their tokens while enhancing functionality.  In a successful upgrade, \TRU (resp.\ \CPU) gradually phases out of circulation in favor of $\TRU'$ ($\CPU'$), while \TRU and \CPU remain useable on both legacy and upgraded systems.  When for example ``$\CPU' = \GPU$'' represents a distinct hardware application, $\CPU$ and $\CPU'$ might incomparably coexist with distinct pricing contracts.  For additional flexibility, one can even transfer \DAO tokens, and hence network effects, across blockchains via a two-way peg \cite{TSB18} or similar mechanism.

\paragraph{Practical example.}
Concurrently with the production of this paper, Harz, Gudgeon, Gervais, and Knottenbelt published a description of a reputation-based scheme called Balance which dynamically adjusts cryptocurrency deposits~\cite{HGGK19}.  The authors mention that the Truebit protocol could benefit from including this design feature.  The upgrade mechanism above incentivizes developers to experiment with Balance's adaptations to Truebit's staking mechanism.

\section{Aleatorics} \label{sec:aleatorics}

We conclude with two upgrade candidates for the Truebit protocol (Section~\ref{sec:review}) which crucially interface Verifiers with randomness and private tasks.

\subsection{Martingale strategy}
Let us examine the effect of private tasks on network security.  Recall  that both the Task Giver's cost and Verifier's reward depends on the computational complexity of the underlying task~(Section~\ref{sec:pricingrewards}, \cite[Section~5.4]{TR17}), and consider a Task Giver who submits a private task which costs 1~\CPU followed by a private task which costs 2~\CPU, then 4~\CPU, and so on doubling the cost each time.  Since the task is private, the Task Giver can verify it himself without concern for splitting rewards with other Verifiers. The Task Giver can even provide bogus answers through a Solver Sybil without spending real compute cycles or risk of getting caught.  Eventually a forced error occurs, at which point the Task Giver more than recoups his expenses for this sequence by challenging the forced error with a payout proportional to the complexity of the last task times the reciprocal of the forced error rate.

The forced error method saves \emph{gas}, or execution resources from the underlying blockchain network, in comparsion to distributing rewards after each task.  Forced errors, however, must be sufficiently infrequent so as to make execution of the martingale strategy above prohibitively expensive because the martingale vulnerability exists whenever some rational attacker has enough capital to sustain the rounds of doubling capital.  On the other hand, if gas is not a concern, the protocol can avoid Martingale attacks entirely by setting the forced error rate equal to~1, that is, imposing a forced error on every single task.

\subsection{Random selection} \label{sec:random}
In parallel to the traditional exponential dropoff payout for Verifiers discussed in Section~\ref{sec:review}, we add a linear reward payout scheme in the spirit of~\cite{KR18}.  In addition to facilitating predictable reward payouts, this gas-efficient approach also improves security.  This \emph{random selection}, which we outline below, is uniformly compatible with private tasks.  The Truebit protocol requires that tasks pass verification for both traditional exponential dropoff and linear random selections.
\begin{enumerate}
\item When a Task Giver submits a task requesting~$k$ Verifiers, priced proportionally, interested Solvers and Verifiers each execute it locally.

\item The protocol then randomly selects, according to the mining nonce, $k+1$ Solver/Verifiers among those with registered deposits.  Each selected Solver/Verifier has the opportunity to commit the hash of a solution plus some committed random bits.  The latter random bits prevent lazy copying of others' solutions.

\item Among the responding subset from Step~2, the protocol randomly selects a Solver, and the remaining nodes become Verifiers.

\item Solvers and Verifiers reveal their solutions.  Verification games ensue in case of disagreement until either one Verifier wins or all Verifier challenges have been refuted.  Unselected Solvers and Verifiers who submit solutions suffer penalties.

\item Absent any dispute, Solvers and Verifiers split the reward equally.
\end{enumerate}
The protocol never penalizes selected Solver/Verifiers who cannot access the input data.  Indeed, participation is optional in Step~2.  Moreover, The construction achieves gas efficiency since only Solver/Verifiers selected in Step~2 who commit solutions actually interact with the blockchain.  The protocol guarantees a predictable reward due to a bounded number of selected participants, and the only real cost for non-selected Solvers and Verifiers are watching events in a contract and performing off-chain computations.

Dandelion, Sam Moelius, and the Arbitrum paper~\cite{KGCWF18} each described variants of the following vulnerability on Truebit's exponential dropoff scheme.  An attacker convincingly, publicly declares to all Verifiers,
\begin{quote}
``I will challenge all tasks of the form \textit{ABC} roughly 100 times in case of a forced error.  Don't bother verifying these, as the reward is negligible.''
\end{quote}
The attacker's goal is to get bogus answers onto the blockchain, and this becomes easier when he convinces other Verifiers not to participate.  Random selection makes such an attack less effective, even under the assumption of entirely rational miners, as the attacker's Sybil attack has no effect on the incentives of chosen Solver/Verifiers.

\paragraph{Acknowledgements.} We thank Ian Kovalenko for help with the \DAO token design.

\bibliographystyle{plain}
\bibliography{stableCPU}

\end{document}